\documentclass{article}
\usepackage[affil-it]{authblk}
\usepackage{amsfonts}
\usepackage{amsmath}
\usepackage{amssymb}
\usepackage{graphicx}
\usepackage{times}
\usepackage{textcomp}
\usepackage{booktabs}
\usepackage{multirow}
\usepackage{sidecap}
\usepackage{tabularx}
\usepackage{subfig}
\usepackage{float}
\usepackage{amsthm}

\usepackage{etoolbox}
\patchcmd{\thebibliography}
  {\settowidth}
  {\setlength{\itemsep}{0pt plus 0.1pt}\settowidth}
  {}{}
\apptocmd{\thebibliography}
  {\small}
  {}{}

\newcommand{\etal}{\text{et al. }}
\newcommand{\Ded}{{D_\text{Edit}}}
\newcommand{\Dsg}{{D_\text{Sg}}}
\newcommand{\Dbot}{{D_\text{Bot}}}
\newcommand{\Dbotsg}{{D_\text{BotSg}}}
\newcommand{\M}{{\mathcal{M}}}
\newcommand{\Mtai}{{\mathcal{M}^\text{Tai}}}
\newcommand{\Msg}{{\mathcal{M}^\text{Sg}}}
\newcommand{\Mbot}{{\mathcal{M}^\text{Bot}}}
\newcommand{\Mbotsg}{{\mathcal{M}^\text{BotSg}}}

\newtheorem{theorem}{Theorem}

\newtheorem{lemma}{Lemma}
\theoremstyle{definition}
\newtheorem{definition}{Definition}

\begin{document}
\title{Improved Methods for Computing Distances between Unordered Trees Using Integer Programming}

\author[1]{Eunpyeong Hong\thanks{ephong93@gmail.com}}
\author[2]{Yasuaki Kobayashi\thanks{kobayashi@iip.ist.i.kyoto-u.ac.jp}}
\author[2]{Akihiro Yamamoto\thanks{akihiro@i.kyoto-u.ac.jp}}
\affil[1]{Kyoto University}
\affil[2]{Graduate Institute of Informatics, Kyoto University}
\date{}
\maketitle
\begin{abstract}
Kondo \etal (DS 2014) proposed methods for computing distances between unordered rooted trees by transforming an instance of the distance computing problem into an instance of the integer programming problem. They showed that the tree edit distance, segmental distance, and bottom-up segmental distance problem can be respectively transformed into an integer program which has $O(nm)$ variables and $O(n^2m^2)$ constraints, where $n$ and $m$ are the number of nodes of input trees. In this work, we propose new integer programming formulations for these three distances and the bottom-up distance by applying dynamic programming approach. We divide the tree edit distance problem into $O(nm)$ subproblems each of which has only $O(n + m)$ constraints. For the other three distances, each subproblem can be reduced to a maximum weighted matching problem in a bipartite graph which can be solved in polynomial time. In order to evaluate our methods, we compare our method to the previous one due to Kondo \etal The experimental results show that the performance of our methods have been improved remarkably compared to that of the previous method.
\end{abstract}
\section{Introduction}
In machine learning applications, it is important to compare (dis)similarities between tree-structured data such as XML and RNA secondary structures. 
There are many measures of similarities between two trees. The tree edit distance \cite{tai} is one of the most widely used measures, which is defined as the minimum cost of edit operations to 
transform a tree into another. It is equivalent to finding the maximum cost of a Tai mapping between two trees.
However, the tree edit distance may not be appropriate to use in some applications where structure-sensitivity is required. 
In this context, many variants of Tai mapping have been proposed (see \cite{kuboyama}, for example). 
In this study, four measures are covered including the edit distance, segmental distance \cite{seg},
bottom-up segmental distance \cite{seg} and bottom-up distance \cite{bot}.

It is known that most of distances between ordered rooted trees can be computed in polynomial time.
For example, Tai \cite{tai} showed that the tree edit distance between ordered rooted trees can be computed in $O(n^3m^3)$ time, where $n$ and $m$ are the number of nodes of input trees, 
and Demaine \etal \cite{demaine} improved the running time to $O(nm^2(1+\log\frac{n}{m}))$. However, if input trees are unordered, the problems of computing the above four distances are known to
be not only NP-hard \cite{ted-np-hard}, but also MAX SNP-hard \cite{seg,bot,ted-max-snp-hard}.
Akutsu \etal studied the tree edit distance problem between unordered trees from a theoretical algorithmic perspective.
They gave an approximation algorithm and exact algorithms \cite{akutsuapprox,akutsufpt,akutsu}.
From the practical point of view, many researches have been done so far.
Horesh \etal \cite{horesh} proposed an A$^*$ algorithm to solve this problem for unlabeled unordered trees and Higuchi \etal \cite{higuchi} extended it for labeled trees.
Fukagawa \etal \cite{fukagawa} proposed a method to reduce the edit distance problem into the maximum vertex weighted clique problem which can be solved by an algorithm due to \cite{tomita}. They showed that the clique-based method is as fast as A*-based method. 
Mori \etal \cite{mori} improved it by applying a dynamic programming approach. They showed that their method is faster than the previous clique-based method. 
Kondo \etal \cite{kondo} proposed a method to reduce an instance of the edit distance problem into an instance of integer linear programming (IP) problem
with $O(nm)$ variables and $O(n^2m^2)$ constraints, where $n$ and $m$ are the number of nodes of input trees, respectively.
However, the instance of their IP formulation has a large number of constraints and hence their method may not be applicable to moderate-sized instances.
Although they showed that their method is faster than the clique-based method of Mori \etal \cite{mori} 
when input trees have large degree nodes, their IP-based method is not very effective when input trees have no large degree nodes or the size of tree is large.

An advantage of IP-based method is that we can easily make an IP formulation representing variations of the edit distance by adding some additional constraints.
In fact, Kondo \etal showed IP formulations which represent segmental distance and bottom-up segmental distance by adding appropriate constraints.
Another advantage of this method is that we can use state-of-the-art IP solvers (e.g. Gurobi, CPLEX), which can quickly solve many hard problems.

In this paper, we propose improved methods to compute the edit distance, segmental distance, bottom-up segmental distance and bottom-up distance between unordered rooted trees.
The improvement of computational efficiency is obtained by applying a dynamic programming approach due to \cite{mori}.
However, it is not only sufficient to apply the dynamic programming but it is necessary to use a structural property of rooted trees.
Their dynamic programming with this property allows us to drastically reduce the number of constraints in our IP formulations for the above distances.
For the edit distance problem, our method has to solve $O(nm)$ subproblems each of which has only $O(n + m)$ constraints. 
For the other distances, each subproblem except the problem of combining the solutions of subproblems can be reduced to the maximum weighted matching problem in a bipartite graph,
which can be solved in polynomial time using the Hungarian method \cite{hungarian}.

The rest of the paper is organized as follows. 
We give notations and preliminary results in Sect. \ref{sec:preli} and briefly explain the previous method in Sect. \ref{sec:previous}.
In Sect. \ref{sec:ours}, we introduce our new methods. 
In order to evaluate our methods, we implemented previous and our methods and conducted experiment
using Glycan dataset \cite{kegg} and CSLOGS dataset \cite{zaki}. The results of our experiments are shown in Sect. \ref{sec:experiment}.
Finally, we conclude our paper with some discussions.
\section{Preliminaries}\label{sec:preli}
Let $T$ be a rooted tree. The root of $T$ is denoted by $r(T)$.
In this paper, we simply write $T$ to represent the set of nodes of $T$.
For $x, y \in T$, $x \le y$ means that $x$ is on the unique path between the root and $y$.
If $x \le y$ and $x \neq y$, we write $x < y$ and say that $x$ is an {\em ancestor} of $y$ and $y$ is a {\em descendant} of $x$.
It is easy to see that the relation $\le$ is a partial order on $T$.
A {\em parent} of $x$, denoted by $p(x)$, is the closest ancestor of $x$.
The {\em children} of $x$, denoted by $C(x)$, is the set of the closest nodes to $x$ among the all descendants of $x$.
We call the number of children of $x$ the {\em degree} of $x$. A node $x$ is called a leaf if it has no children. The set of all leaves of a tree $T$ is denoted by $L(T)$.
Nodes $x$ and $y$ are {\em siblings} if they has the same parent.
A tree is called unordered tree if there is no order between siblings.
Let $\Sigma$ be a finite alphabet and $l_T: T \rightarrow \Sigma$ a labeling function.
A tuple $(T,l_T)$ is called a {\em labeled} tree.
For $x \in T$, we use $T(x)$ to denote the subtree of $T$ rooted at $x$.
For notational convenience, we simply write $T - x$ to denote the subgraph of $T$ obtained by removing a node $x$.

\subsection{Tree Edit Distance}
The tree edit distance between two trees is defined as the minimum cost of {\em edit operations} to transform a tree into another.
\begin{definition}[Edit Operations]
Let T be a tree. {\em Edit operations} on T consist of the following three operations.
\end{definition}
\begin{description}
\item [Substitution] Replace the label of a node in $T$ with a new label.
\item [Deletion] Delete a non-root node $t$ of $T$, making all children of $t$ be the children of $par(t)$.
\item [Insertion] Insert a new node $t$ as a child of some node $v$ in $T$, making some children of $v$ be the children of $t$.
\end{description}
Let $\Sigma_{\varepsilon} = \Sigma \cup \{\varepsilon\}$, where $\varepsilon$ is a blank symbol not in $\Sigma$.
In order to describe costs on edit operations, we denote each of the edit operations by a pair in $\Sigma_\varepsilon \times \Sigma_\varepsilon\setminus \{(\varepsilon,\varepsilon)\}$.
Substituting a node labeled with $a$ by another node labeled with $b$ is denoted by $(a, b)$. Inserting a node labeled with $b$ is denoted by $(\varepsilon, b)$. Deleting a node labeled with $a$ is denoted by $(a, \varepsilon)$.
Let $d : \Sigma_\varepsilon \times \Sigma_\varepsilon \setminus \{(\varepsilon, \varepsilon)\}\rightarrow \mathbb{R^+}$ be a cost function on edit operations and assume, in this paper, that $d$ is a metric. In the following, we simply write $d(x, y)$ for $(x, y)\in S\times T$
to represent $d(l_1(x), l_2(y))$, where $l_1$ and $l_2$ are labeling functions on two trees $T_1$ and $T_2$, respectively.

Let $E=\langle e_1, e_2, \ldots, e_t\rangle$ be a sequence of edit operations, where $e_i = (a_i, b_i)$ for $a_i, b_i \in \Sigma_{\varepsilon}$.
The cost of the sequence is defined as $ cost(E) = \sum_{1 \le i \le t}d(e_i)$.
\begin{definition}[Tree Edit Distance \cite{tai}]
Let $T_1$ and $T_2$ be trees and $\mathcal{E}(T_1, T_2)$ be the set of all sequences of edit operations which transform $T_1$ into $T_2$.
The {\em tree edit distance} between $T_1$ and $T_2$, denoted by $\Ded(T_1, T_2)$, is defined as $\Ded(T_1, T_2) = \min_{E\in \mathcal{E}(T_1, T_2)}cost(E)$
\end{definition}
A {\em mapping} between $T_1$ and $T_2$ is a subset of $T_1 \times T_2$.
The set of nodes that belongs to a mapping $M$ is denoted by $V(M)$.
Tai \cite{tai} gave a combinatorial characterization of the tree edit distance by means of a mapping, which is called a {\em Tai mapping}.
\begin{definition}[Tai Mapping \cite{tai}]
Let $T_1$ and $T_2$ be trees. A mapping $M$ is  called a {\em Tai mapping} if it satisfies
the following constraints for every $(x, y),(x', y')$ in $M$:
\begin{description}
\item [One-to-one correspondence : ] $x = x' \Leftrightarrow y = y'$,
\item [Preserving ancestor-descendant relationship: ] $x < x' \Leftrightarrow y < y'$.
\end{description}
\end{definition}
The cost of a Tai mapping $M$ is defined as 
\begin{equation*}
\displaystyle cost(M) = \sum_{(x, y) \in M}d(x, y) + 
\sum_{x \in T_1 \setminus V(M)}d(x, \varepsilon) +
\sum_{y \in T_2 \setminus V(M)} d(\varepsilon, y).
\end{equation*}
Let $\Mtai(T_1, T_2)$ be the set of all Tai mappings between $T_1$ and $T_2$. Tai \cite{tai} showed the following theorem.
\begin{theorem}[\cite{tai}]
For two trees $T_1$ and $T_2$, $\displaystyle\Ded(T_1, T_2) = \min_{M \in \Mtai(T_1, T_2)} cost(M)$.
\end{theorem}

\subsection{Variants of Edit Distance}\label{sec:variants}
The tree edit distance is one of the most widely used to measure a similarity between two trees. However, it may not be appropriate for some applications  because one may need a distance on which some specific structure of trees is reflected.
Many variants of the tree edit distance have been proposed in the literature \cite{seg,bot}.
We work on the following three variants, which are defined by mappings rather than edit operations.
\begin{definition}[Segmental Mapping \cite{seg}]
Let $T_1$ and $T_2$ be trees. A Tai mapping $M$ between $T_1$ and $T_2$ is called a {\em segmental mapping}
if for any $(x, y), (x', y') \in M$ with $x < x'$ and $y < y'$, $(p(x'), p(y')) \in M$.
\end{definition}
\begin{definition}[Bottom-up Segmental Mapping \cite{seg}]
Let $T_1$ and $T_2$ be trees. A segmental mapping $M$ between $T_1$ and $T_2$ is called a {\em bottom-up segmental mapping} 
if for any $(x, y) \in M$, there is $(x', y') \in M$ such that $x', y'$ are leaves with $x \le x'$ and $y \le y'$.
\end{definition}
\begin{definition}[Bottom-up Mapping \cite{bot}]\label{def:bot}
Let $T_1$ and $T_2$ be trees. A Tai mapping $M$ between $T_1$ and $T_2$ is called a {\em bottom-up mapping}
if for any $(x, y) \in M$, the submapping obtained from $M$ by restricting to $C(x) \times C(y)$ forms a bijection between $C(x)$ and $C(y)$.
\end{definition}
Let us note that the condition in Definition~\ref{def:bot} can be restated in the following way:
M is a bottom-up mapping if for any $(x, y) \in M$, the submapping obtained from $M$ by restricting to $T_1(x) \times T_2(y)$ is an isomorphism mapping, ignoring the label information.
\begin{definition}[\cite{seg,bot}]
Let $T_1$ and $T_2$ trees. Denote the sets of all possible segmental mappings, bottom-up segmental mappings, and bottom-up mappings between $T_1$ and $T_2$
by $\Msg(T_1,T_2),\Mbotsg(T_1,T_2)$, and $\Mbot(T_1,T_2)$, respectively.
The {\em segmental distance}, {\em bottom-up segmental distance}, and {\em bottom-up distance} between $T_1$ and $T_2$, which are denoted by $\allowbreak \Dsg(T_1,T_2), \Dbotsg(T_1,T_2)$,
and $\Dbot(T_1,T_2)$ respectively, are defined as follows:
\begin{align}
&\Dsg(T_1, T_2) = \min_{M\in \Msg(T_1, T_2)} cost(M)\nonumber\\
&\Dbotsg(T_1, T_2) = \min_{M \in \Mbotsg(T_1, T_2)} cost(M)\nonumber\\
&\Dbot (T_1, T_2) = \min_{M \in \Mbot(T_1, T_2)} cost(M).\nonumber
\end{align}
\end{definition}

\section{Previous Method \cite{kondo}}\label{sec:previous}
In the rest of this paper, fix input trees $T_1$ and $T_2$, and
let $n = |T_1|$ and $m = |T_2|$.
Kondo \etal \cite{kondo} proposed an integer linear programming formulation for the tree edit distance.
For the tree edit distance between $T_1$ and $T_2$, we introduce a binary variable $m_{x, y}$ for every $(x, y) \in T_1 \times T_2$
which takes value 1 if and only if $(x, y) \in \Mtai(T_1,T_2)$.
Then, we can reformulate the cost of a Tai mapping $M$ as:
\begin{eqnarray*}
cost(M) &=& \sum_{(x, y) \in M} d(x, y)+\sum_{x \in T_1 \setminus V(M)} d(x, \varepsilon) + 
\sum_{y \in T_2 \setminus V(M)}d(\varepsilon, y) \\
&=&\sum_{(x, y) \in T_1 \times T_2} d(x, y) m_{x, y} + \sum_{x \in T_1}d(x, \varepsilon) \left\{1 - \sum_{y \in T_2}m_{x,y} \right\}
+ \sum_{y \in T_2}d(\varepsilon, y)\left\{1 - \sum_{x \in T_1}m_{x, y}\right\}\\
&=&\sum_{(x, y)\in T_1 \times T_2}\left\{ d(x,y)-d(x,\varepsilon)-d(\varepsilon,y) \right\}m_{x, y}+\sum_{x \in T_1}d(x,\varepsilon)+
\sum_{y \in T_2} d(\varepsilon, y).
\end{eqnarray*}
The two constraints of Tai mapping are directly formulated as the following inequalities:
\begin{eqnarray*}
\sum_{y \in T_2}m_{x, y}\le 1 &&\text{ for all } x \in T_1,\\
\sum_{x \in T_1}m_{x, y}\le 1 &&\text{ for all } y \in T_2,\\
m_{x, y}+m_{x', y'} \le 1 &&\text{ for all }(x, y), (x', y')\in T_1 \times T_2 \text{ s.t. } x < x' \not\Leftrightarrow y < y'.
\end{eqnarray*}
The first two constraints are equivalent to the one-to-one correspondence of Tai mapping. It means that for any node $x \in T_1$ (resp. $y \in T_2$),
at most one node of $T_2$ (resp. $T_1$) is allowed to be paired.
The third constraint is equivalent to the ancestor-descendant preservation.
It means that for any two pairs which do not preserve the ancestor-descendant relationship, both of them cannot be included in $M$ simultaneously.
This formulation contains $O(nm)$ variables and $O(n^2m^2)$ constraints.

Kondo \etal also gave IP formulations for the segmental distance and bottom-up segmental distance.
These distances can be formulated by imposing additional constraints on the above formulation.
In regard of the segmental mapping, the constraints of segmental mapping can be represented as follows:
\begin{eqnarray*}
\begin{array}{l}
m_{x,y}+m_{x',y'}\le m_{p(x'),p(y')}+1, \text{ for all }(x,y),(x',y')\in T_1 \times T_2 \text{ s.t. }x<x' \text{ and }y<y'.
\end{array}
\end{eqnarray*}
The constraints of bottom-up segmental mapping can also be represented as follows:
\begin{eqnarray*}
\begin{array}{l}\displaystyle
m_{x,y}\le \sum_{\substack{x' \in L(T_1(x)),\\ y' \in L(T_2(y))}}m_{x', y'},
\text{ for all }(x,y)\in T_1\times T_2 \text{ s.t. } x\notin L(T_1) \text{ and } y \notin L(T_2).
\end{array}
\end{eqnarray*}
The above two formulations also contain $O(nm)$ variables and $O(n^2m^2)$ constraints.

\section{Improved Method}\label{sec:ours}
\subsection{Improved Method for Tree Edit Distance}
In this section, we propose a new IP formulation for the edit distance problem by combining a dynamic programming approach due to \cite{mori}.
The dynamic programming computes a minimum cost Tai mapping $M_{x, y}$ between $T_1(x)$ and $T_2(y)$ with $(x, y) \in M_{x, y}$ for $(x,y) \in T_1\times T_2$ in a bottom-up manner.
Once we have the solutions for all pairs $(x, y) \in T_1 \times T_2$, we can construct a minimum cost Tai mapping between $T_1$ and $T_2$.

First, we modify the objective function
\begin{equation*}
\text{minimize}\sum_{(x, y)\in T_1 \times T_2}\{d(x, y) - d(x, \varepsilon) - d(\varepsilon, y)\}m_{x, y}+\sum_{x \in T_1}d(x,\varepsilon) + \sum_{y \in T_2}d(\varepsilon, y)
\end{equation*}
to
\begin{equation*}
\text{maximize}\sum_{(x, y)\in T_1 \times T_2}w_{x, y}m_{x, y},
\end{equation*}
where $w_{x, y} = d(x, \varepsilon) + d(\varepsilon, y) - d(x, y)$.
This modification is valid since the second and third terms do not affect the minimization.

Since the solution of our subproblem for $T_1(x)$ and $T_2(y)$ must contain the root pair $(x, y)$,
the objective function on the input trees $T_1(x)$ and $T_2(y)$ can be represented as
\begin{equation}
\begin{array}{lll}
\text{maximize}  \displaystyle \sum_{(x', y') \in (T_1(x) - x) \times (T_2(y) - y)} w_{x', y'}m_{x', y'}+w_{x, y}\label{eq:obj}.
\end{array}
\end{equation}
We denote by $W_{x, y}$ the maximum value of (\ref{eq:obj}).
If at least one of $x$ and $y$ is a leaf, $W_{x, y} = w_{x, y}$.
Thus, in the following, we assume that neither $x$ nor $y$ is a leaf. 
The idea for our dynamic programming is that $W_{x, y}$ can be recursively computed from the values $W_{x', y'}$ for $x < x'$ and $y < y'$.
To be precise, let $\M^*(T_1(x), T_2(y))$ be the set of all Tai mappings $M$ between $T_1(x)$ and $T_2(y)$ such that $(x, y) \notin M$ and
both $T_1 \cap V(M)$ and $T_2 \cap V(M)$ are antichains in $(T_1(x), \le)$ and $(T_2(y), \le)$, respectively.
For a Tai mapping $M$, we let $w(M)$ and $W(M)$ to denote $\sum_{(x, y) \in M}w_{x, y}$ and $\sum_{(x, y) \in M}W_{x, y}$, respectively.
The following lemma is a key ingredient of our formulation.
\begin{lemma}\label{lem:key}
    $\displaystyle W_{x, y} = \max_{M \in \M^*(T_1(x), T_2(y))} W(M) + w_{x, y}$.
\end{lemma}
\begin{proof}
We first show that the left-hand side is at most the right-hand side.
Let $M$ be a Tai mapping between $T_1(x)$ and $T_2(y)$ with $(x, y) \in M$.
Then, $M$ can be uniquely decomposed into $(x, y), M_{x_1, y_1}, M_{x_2, y_2}, \ldots, M_{x_k, y_k}$ such that
for any $1 \le i \le k$, $M_{x_i, y_i}$ is a Tai mapping between $T_1(x_i)$ and $T_2(y_i)$ with $(x_i, y_i) \in M_{x_i, y_i}$ and $\bigcup_{1 \le i \le k}(x_i,y_i) \in \M^*(T_1(x), T_2(y))$.
Such a decomposition can be obtained by choosing minimal node pairs $(x_i, y_i) \in M \setminus \{(x, y)\}$ with respect to $\le$: For any $(x', y') \in M$
either $x_i \le x'$ and $y_i \le y'$, or $x_i$ and $y_i$ are not comparable to $x'$ and $y'$, respectively. 
For each $1 \le i \le k$, we have $w(M_{x_i, y_i}) \le W_{x_i, y_i}$.
Therefore, $W_{x, y} \le \sum_{1\le i\le k}W_{x_i, y_i} + w_{x, y} = \max_{M \in \M^*(T_1(x), T_2(y))} W(M) + w_{x, y}$.

To show the converse, let $M^* \in \M^*(T_1(x), T_2(y))$ be maximizing the right-hand side.
For each $(x', y') \in M^*$, we let $M_{x', y'}$ be a Tai mapping between $T_1(x')$ and $T_2(y')$ such that $W_{x', y'} = w(M_{x', y'})$ and $(x', y') \in M_{x', y'}$.
Since $T_1(x) \cap V(M^*)$ and $T_2(y) \cap V(M^*)$ are antichains,
$\bigcup_{(x', y') \in M^*} M_{x', y'} \cup \{(x, y)\}$ is a Tai mapping between $T_1(x)$ and $T_2(y)$.
Therefore, we have $\max_{M \in \M^*(T_1(x), T_2(y))} W(M) + w_{x, y} \le \sum_{(x', y') \in M^*} w(M_{x', y'}) + w_{x, y} \le W_{x, y}$ and hence the lemma holds.
\qed\end{proof}

By Lemma~\ref{lem:key}, our problem is to maximize
\begin{equation*}
\sum_{(x', y') \in M} W_{x', y'}m_{x', y'}+w_{x, y}
\end{equation*}
subject to $M \in \M^*(T_1(x), T_2(y))$.

Mori \etal \cite{mori} reduced the problem of finding a maximum weight Tai mapping in $\M^*(T_1(x), T_2(y))$
to the maximum vertex weight clique problem, which corresponds to the maximum weight independent set problem on complement graphs.
Their reduction can be interpreted as the following constraint:
\begin{eqnarray*}
\begin{array}{l}
m_{x', y'} + m_{x'',y''} \le 1 \text{ for all }(x', y'), (x'', y'')\in T_1(x) \times T(y) \text{ s.t. }x' < x'' \text{ or } y' < y''.
\end{array}
\end{eqnarray*}
However, this formulation contains $\Omega(n^2m^2)$ constraints.

In order to reduce the number of constraints, we will exploit a structure of rooted trees.
For a node $x \in T$ and a leaf $l \in L(T(x))$, let $P^T_{xl}$ be the unique path between $x$ and $l$ in $T$.
Then, for any $M \in \M^*(T_1(x), T_2(y))$ and any $l \in L(T_1(x))$ (resp. $l \in L(T_2(y))$), at most one node of $P^{T_1}_{xl}$ (resp. $P^{T_2}_{yl}$) can be chosen in $M$, that is,
\begin{equation*}
\begin{array}{ll}
    \displaystyle\sum_{x' \in P^{T_1}_{xl} - x}\sum_{y' \in T_2(y)}m_{x', y'} \le 1 & \text{ for all } l \in L(T_1(x)), \\
    \displaystyle\sum_{y' \in P^{T_2}_{yl} - y}\sum_{x' \in T_1(x)}m_{x', y'} \le 1 & \text{ for all } l \in L(T_2(y)).
\end{array}
\end{equation*}
This is formalized by the following lemma.

\begin{lemma}\label{lem:ip}
Let $x \in T_1$ and $y \in T_2$.
Then, $W_{x, y}$ can be computed by the following IP.
\begin{eqnarray*}
\begin{array}{lll}
\text{maximize} & \displaystyle\sum_{x' \in T_1(x) - x, y' \in T_2(y) - y}W_{x', y'}m_{x', y'} + w_{x, y}&\\
\text{subject to} & \displaystyle\sum_{x' \in P^{T_1}_{xl} - x}\sum_{y' \in T_2(y) - y}m_{x', y'} \le 1 & \text{ for all } l \in L(T_1(x))\\
& \displaystyle\sum_{y' \in P^{T_2}_{yl} - y}\sum_{x' \in T_1(x)}m_{x', y'} \le 1 & \text{ for all } l \in L(T_2(y))\\
& m_{x', y'} \in \{0, 1\}& \text{ for all } x' \in T_1(x) - x, y' \in T_2(y) - y.
\end{array}
\end{eqnarray*}
\end{lemma}
\begin{proof}
By Lemma~\ref{lem:key}, it suffices to prove that $M = \{ (x', y') : x' \in T_1(x), y' \in T_2(y), m_{x', y'} = 1\}$ is in $M^*(T_1(x), T_2(y))$
if and only if $m_{*,*}$ is a feasible solution.

Suppose first that $M \in M^*(T_1(x), T_2(y))$.
Since $T_1(x) \cap V(M)$ forms an antichain in $(T_1, \le)$, $M$ has at most one node in $P_{xl}^{T_1}$ for each $l \in L(T_1(x))$. 
Therefore, binary variables $m_{x', y'}$ do not violate the first type constraints.
A symmetric argument for $T_2(y) \cap V(M)$ implies that $m_{*,*}$ is a feasible solution for the IP.

Suppose, for contradiction, $m_{*,*}$ is a feasible solution and there are $(x', y'), (x'', y'')$ in $M$ that violate the condition of $M^*(T_1(x), T_2(y))$.
There are two possibilities: $(x', y')$ and $(x'', y'')$ violate the one-to-one correspondence of Tai mapping or
at least one of $x' < x''$ or $y' < y''$ holds.
For the former case, assume without loss of generality that $x' = x''$ and $y' \neq y''$.
In this case, the pairs contribute at least two to a constraint for each $l \in T_1(x')$, which contradict the feasibility of $m_{*,*}$.
For the latter case, assume without loss of generality that $x' < x''$.
In this case, there is a path $P^{T_1}_{xl} - x$ that contains both $x'$ and $x''$.
The pairs contribute at least two to a constraint for such $l \in L(T_1(x))$ , which also contradict the feasibility of $m_{*,*}$.
Therefore, the lemma holds.
\qed\end{proof}

For $x \in T_1$ and $y \in T_2$, we can compute $W_{x, y}$ by using the formulation of Lemma~\ref{lem:ip}.
The remaining task is to compute $\Ded(T_1, T_2)$ from the values $W_{x, y}$.
\begin{theorem}\label{thm:common-ip}
Let $opt$ be the optimal value of the following IP.
Then, $\Ded(T_1, T_2) = \displaystyle\sum_{x \in T_1}d(x, \varepsilon) + \sum_{y \in T_2} d(\varepsilon, y) - opt$.
\begin{eqnarray*}
\begin{array}{lll}
\text{maximize} & \displaystyle\sum_{x \in T_1, y \in T_2}W_{x, y}m_{x, y}&\\
\text{subject to} & \displaystyle\sum_{x \in P^{T_1}_{r(T_1)l}}\sum_{y \in T_2}m_{x, y} \le 1 & \text{ for all } l \in L(T_1)\\
& \displaystyle\sum_{y \in P^{T_2}_{r(T_2)l}}\sum_{x \in T_1}m_{x, y} \le 1 & \text{ for all } l \in L(T_2)\\
& m_{x, y} \in \{0, 1\}& \text{ for all } x \in T_1, y \in T_2.
\end{array}
\end{eqnarray*}
\end{theorem}

The proof of theorem~\ref{thm:common-ip} is analogous to those of Lemma~\ref{lem:key} and \ref{lem:ip}.
Our method has $O(nm)$ subproblems. Each subproblem, however, contains $O(nm)$ variables and only $O(n + m)$ constraints.

\subsection{Improved Methods for Variants of Edit Distance}
\begin{figure}[htb]
\includegraphics[width=\textwidth]{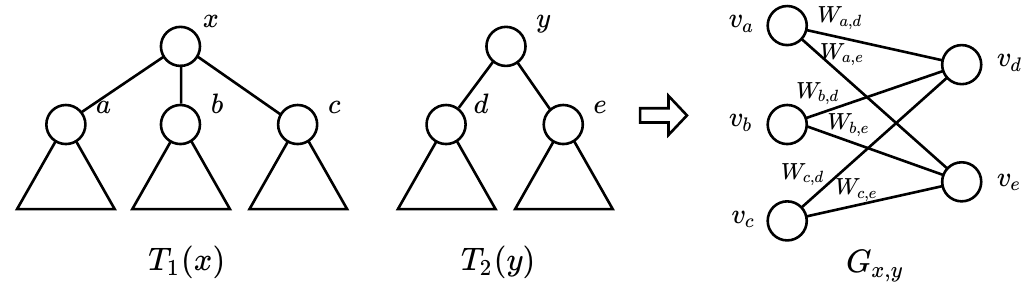}
\caption{The figure illustrates the reduction from the maximum segmental mapping problem to the maximum matching problem in a bipartite graph.}
\label{fig:reduction}
\end{figure}
As the edit distance was computed in the previous section, the other distances can also be computed in the same manner: For each $x \in T_1$ amd $y \in T_2$, compute $W_{x, y}$, and
then combine the solutions $W_{x, y}$ of subproblems as in Theorem~\ref{thm:common-ip}.

\subsubsection{Segmental Distance}
Let $x$ and $y$ be nodes of two trees $T_1$ and $T_2$, respectively. 
We denote here by $W_{x, y}$ the maximum weight, that is the maximum value of (\ref{eq:obj}), of segmental mappings $M_{x, y}$ between $T_1(x)$ and $T_2(y)$ with $(x, y) \in M_{x, y}$.
If either $x$ or $y$ is a leaf, we have $W_{x, y} = w_{x, y}$.
Thus, we suppose otherwise.
Suppose $W_{x',y'}$ have already computed for each $(x', y') \in (T_1(x) \times T_2(y)) \setminus \{(x, y)\}$. 
Observe that for any segmental mapping $M_{x, y}$ with $(x, y) \in M_{x, y}$, a child of $x$ must be paired with a child of $y$ in $M_{x, y}$.
Moreover, if a descendant $x'$ of $x$ that is not a child of $x$ is in $V(M_{x, y})$, the child of $x$ that is an ancestor of $x'$ must be in $V(M_{x, y})$.
These observations imply that  $M_{x, y}$ can be constructed by a union of mappings $M_{x', y'}$ for $x' \in C(x)$ and $y' \in C(y)$,
where $M_{x', y'}$ is a mapping between $T_1(x')$ and $T_2(y')$ with $(x', y') \in M_{x', y'}$. 
Therefore, in order to compute $W_{x, y}$, we construct a bipartite graph $G_{x, y}$ as follows.
For each $z \in C(x) \cup C(y)$, we create a vertex $v_z$ and for each $x' \in C(x)$ and $y' \in C(y)$, add an edge between $v_{x'}$ and $v_{y'}$ whose weight equals $W_{x', y'}$ as in Fig.~\ref{fig:reduction}.
It is well-known that a maximum weight bipartite matching can be solved in polynomial time using Hungarian method~\cite{hungarian}. 

When $W_{x, y}$ is computed for each $x \in T_1$ and $y \in T_2$, we can compute the segmental distance between $T_1$ and $T_2$ by Theorem~\ref{thm:common-ip}.

\subsubsection{Bottom-up Segmental Distance}
Because any bottom-up segmental mapping is a segmental mapping, the above observations also hold and each subproblem can be reduced to a maximum weight matching problem in a bipartite graph as well.
The only difference from the case of segmental distance is that every segment must include at least one leaf. 
To this end, we need to exclude the following two cases from our solution.
If exactly one of $x$ and $y$ is a leaf, then $W_{x, y}$ must be zero since $(x, y)$ violates the condition of bottom-up segmental mapping.
The other case is that neither $x$ nor $y$ is a leaf and the solution of the maximum weight matching equals zero.
This implies that an optimal mapping between $T_1(x)$ and $T_2(y)$ consists of a single pair $(x, y)$, which also violates the condition of bottom-up segmental mapping.
Therefore, we set $W_{x, y} = 0$ in this case.

\subsubsection{Bottom-up Distance}
First, we propose a naive IP formulation for computing bottom-up distance. A straightforward implication from Definition~\ref{def:bot} is that if $(x, y) \in M$, the mapping 
between $C(x)$ and $C(y)$ must be a bijection.
The formulation can be obtained from that of Tai mapping by adding the following constraints:
\begin{eqnarray*}
\begin{array}{l}\displaystyle
m_{x, y}\le \sum_{y' \in C(y)} m_{x', y'}
\text{ for all }(x,y)\in T_1\times T_2 \text{ and for all } x' \in C(x),\\
\displaystyle m_{x, y}\le \sum_{x' \in C(x)} m_{x', y'}
\text{ for all }(x,y)\in T_1\times T_2 \text{ and for all } y' \in C(y).
\end{array}
\end{eqnarray*}
This formulation contains $O(nm)$ variables and $O(n^2m^2)$ constraints.

Since bottom-up mapping is a subclass of bottom-up segmental mapping, we can apply the above technique as well. 
All we have to do is to consider the case when two trees $T_1(x)$ and $T_2(y)$ are structurally isomorphic.
Thus, for $x \in T_1$ and $y\in T_2$, we set $W_{x,y} = 0$ if two subtrees $T_1(x)$ and $T_2(y)$ are not structurally isomorphic, i.e., they are isomorphic ignoring the labels. 

Our improved methods contain $O(nm)$ subproblems which can be solved in polynomial time.
For combining the solutions of these subproblems, we need to solve an integer program in Theorem~\ref{thm:common-ip}.
Such IPs also have $O(nm)$ variables and $O(n + m)$ constraints.

\section{Experiments}\label{sec:experiment}
To compare the experimental performance of our methods and the previous methods, we applied them to real tree-structured data. We used glycan data obtained from KEGG/Glycan database \cite{kegg} and CSLOGS dataset \cite{zaki} which consists of web log files.
In our experiments, we adopt the {\em unit cost} for the cost function, which is defined as:
\begin{equation*}
d(x, y) =
\left\{
\begin{array}{ll}
    0 & \text{ if } l_1(x) = l_2(y)\\
    1 & \text{ otherwise }
\end{array}
\right..
\end{equation*}

We implemented the previous methods for computing edit distance (IP\_Edit), segmental distance (IP\_Sg), and bottom-up segmental distance (IP\_BotSg) given by Kondo \etal \cite{kondo} and
a naive method for computing bottom-up distance (IP\_Bot) described in the previous section.
We also implemented our methods for computing these four distances (DpIP\_Edit, DpIP\_Sg, DpIP\_BotSg, and DpIP\_Bot).
In addition to the above implementations, we intended to compare our methods with the algorithm due to Mori \etal \cite{mori}.
Their algorithm reduces the tree edit distance problem to the maximum weight clique problem and uses the maximum weight clique algorithm due to \cite{tomita}.
However, the purpose of our experiments is to compare formulations or reductions rather than the performance of specific IP or other solvers.
Therefore, we used an ordinary IP formulation of the maximum weight clique problem instead of the algorithm of \cite{tomita}, which is denoted by IP\_DpClique\_E.

We implemented the methods mentioned above in Java 1.8 combined with IBM ILOG CPLEX 12.7.
We have forced CPLEX to run in sequential mode, setting parameter {\tt IloCplex.IntParam.Threads} to one. Every implementation of the presented methods is also single-threaded. 
The experiments were performed using a computer with 3.7 GHz Quad-Core Intel Xeon E5 and 32 GB RAM, under the Mac OS X.

\subsection{Glycan dataset}

The results for edit distance with Glycan dataset are shown in Table \ref{glycanted}. ``\# of nodes'' in the table means the total number of nodes of two input trees. 
We randomly selected at most 100 input tree pairs from the Glycan dataset for each range of total nomber of nodes.
Avg and t.o. stand for average execution time (in seconds) and the number of instances timed out, respectively. The table shows that DpIP\_Edit is much faster than IP\_Edit. IP\_DpClique\_E is not faster than IP\_Edit when the size of inputs are large, while IP\_DpClique\_E outperforms IP\_Edit when the inputs are small-sized trees. It is shown that DpIP\_Edit also outperforms IP\_DpClique\_E. It implies that it is not sufficient to adopt a dynamic programming aproach for improving on the practical performance, and the revised IP formulation derived from the dynamic programming is of great importance for reducing the running time on the tree edit distance problem. 

Table \ref{glycanseg} shows the results for the variants of edit distance. For segmental distance and bottom-up segmental distance, the proposed methods (DpIP\_Sg and DpIP\_BotSg) finished computing within 1 second while the naive methods (IP\_Sg and IP\_BotSg) take longer than 30 seconds if the total size of input trees is large. For bottom-up distance, the naive method (IP\_Bot) was fast as all instances were computed within 30 seconds. However, our improved method (DpIP\_Bot) is still much faster than the naive method.
\begin{table}[H]
\setlength{\tabcolsep}{3pt}
\footnotesize
\centering
\caption{Experimental results with Glycan for edit distance}
\begin{tabular}{c@{\quad}r@{\quad}r@{\quad}r@{\quad}r@{\quad}r@{\quad}r@{\quad}r}
\toprule
\multirow{2}{*}{\# of nodes}&\multirow{2}{*}{\# of instances}&\multicolumn{2}{c}{IP\_Edit}&\multicolumn{2}{c}{DpIP\_Edit}&\multicolumn{2}{c}{IP\_DpClique\_E}\\
&&avg&t.o.&avg&t.o.&avg&t.o.\\
\midrule
50 - 54&100&2.393&0&0.308&0&0.994&0\\
55 - 59&100&4.661&0&0.417&0&1.576&0\\
60 - 64&88&11.661&6&0.576&0&2.894&0\\
65 - 69&36&17.774&4&0.669&0&3.433&0\\
70 - 74&100&13.209&7&0.654&0&11.799&7\\
75 - 79&29&20.771&9&0.823&0&11.411&7\\
80 - 84&9&18.705&8&1.094&0&14.941&6\\
85 - 89&5&0&5&1.330&0&21.838&3\\
90 - 94&4&0&4&1.442&0&0&4\\
\bottomrule
\end{tabular}
\label{glycanted}
\end{table}

\begin{table}[H]
\setlength{\tabcolsep}{3pt}
\footnotesize
\centering
\caption{Experimental results with Glycan for segmental distance, bottom-up segmental distance, and bottom-up distance}
\begin{tabularx}{\textwidth}{crrr@{\quad}rr@{\quad}rr@{\quad}rr@{\quad}rr@{\quad}rr}
\toprule
\multirow{2}{*}{\# of nodes}&\multirow{2}{*}{\begin{minipage}{0.45in}\# of \\instances\end{minipage}}&\multicolumn{2}{c}{IP\_Sg}&\multicolumn{2}{c}{DpIP\_Sg}&\multicolumn{2}{c}{IP\_BotSg}&\multicolumn{2}{c}{DpIP\_BotSg}&\multicolumn{2}{c}{IP\_Bot}&\multicolumn{2}{c}{DpIP\_Bot}\\
&&\multicolumn{1}{c}{avg}&t.o.&\multicolumn{1}{c}{avg}&t.o.&\multicolumn{1}{c}{avg}&t.o.&\multicolumn{1}{c}{avg}&t.o.&\multicolumn{1}{c}{avg}&t.o.&\multicolumn{1}{c}{avg}&t.o.\\
\midrule
50 - 54&100&5.306&0&0.135&0&1.545&0&0.136&0&0.569&0&0.131&0\\
55 - 59&100&9.070&5&0.135&0&2.539&0&0.139&0&0.785&0&0.131&0\\
60 - 64&88&13.983&41&0.137&0&4.767&0&0.142&0&1.258&0&0.132&0\\
65 - 69&36&23.813&27&0.140&0&6.219&0&0.147&0&1.544&0&0.133&0\\
70 - 74&100&20.408&97&0.145&0&10.252&4&0.150&0&1.453&0&0.134&0\\
75 - 79&29&21.274&27&0.148&0&12.794&5&0.154&0&2.021&0&0.137&0\\
80 - 84&9&0&9&0.152&0&17.606&3&0.160&0&3.002&0&0.137&0\\
85 - 89&5&0&5&0.157&0&29.157&4&0.163&0&3.869&0&0.142&0\\
90 - 94&4&0&4&0.161&0&0&4&0.166&0&4.476&0&0.145&0\\
\bottomrule
\end{tabularx}
\label{glycanseg}
\end{table}
\subsection{CSLOGS Dataset}
We divided CSLOGS dataset into two subsets: SUBLOG3 and SUBLOG49. Every tree in SUBLOG3 (resp. SUBLOG49) is restricted to have the maximum degree at most 3 (resp. 49). We randomly selected at most 100 pairs from each dataset with a specified range of the total number of nodes.

The results of computation for SUBLOG3 are shown in Table \ref{sub3ted} and \ref{sub3seg}. Table \ref{sub49ted} and \ref{sub49seg} shows the results for SUBLOG49. Compared to the results in SUBLOG3, the naive methods (IP\_Edit, IP\_Sg, IP\_BotSg, and IP\_Bot) in SUBLOG49 works faster. 
This property is what has been observed in the previous work by Konto et al. In regard of IP\_DpClique\_E, it outperforms IP\_Edit when the degrees of trees are small, though their performances are scarcely different with high-degree inputs.
\begin{table}[H]
\setlength{\tabcolsep}{3pt}
\footnotesize
\centering
\caption{Experimental results with SUBLOG3 for edit distance}
\begin{tabular}{c@{\quad}r@{\quad}r@{\quad}r@{\quad}r@{\quad}r@{\quad}r@{\quad}r}
\toprule
\multirow{2}{*}{\# of nodes}&\multirow{2}{*}{\# of instances}&\multicolumn{2}{c}{IP\_Edit}&\multicolumn{2}{c}{DpIP\_Edit}&\multicolumn{2}{c}{IP\_DpClique\_E}\\
&&avg&t.o.&avg&t.o.&avg&t.o.\\
\midrule
50 - 54&100&2.478&0&0.435&0&3.853&0\\
55 - 59&100&3.892&0&0.510&0&5.393&2\\
60 - 64&100&6.641&0&0.633&0&8.243&17\\
65 - 69&100&9.921&1&0.760&0&7.191&34\\
70 - 74&100&15.077&9&0.917&0&8.244&44\\
75 - 79&100&16.534&29&1.112&0&6.352&47\\
80 - 84&100&19.024&45&1.247&0&5.144&44\\
85 - 89&100&21.249&70&1.449&0&4.711&48\\
90 - 94&100&23.946&91&1.872&0&6.863&59\\
95 - 99&100&26.599&92&2.136&0&7.971&61\\
\bottomrule
\end{tabular}
\label{sub3ted}
\end{table}
\begin{table}[H]
\setlength{\tabcolsep}{3pt}
\footnotesize
\centering
\caption{Experimental results with SUBLOG3 for segmental distance, bottom-up segmental distance and bottom-up distance}
\begin{tabularx}{\textwidth}{crrr@{\quad}rr@{\quad}rr@{\quad}rr@{\quad}rr@{\quad}rr}
\toprule
\multirow{2}{*}{\# of nodes}&\multirow{2}{*}{\begin{minipage}{0.45in}\# of \\instances\end{minipage}}&\multicolumn{2}{c}{IP\_Sg}&\multicolumn{2}{c}{DpIP\_Sg}&\multicolumn{2}{c}{IP\_BotSg}&\multicolumn{2}{c}{DpIP\_BotSg}&\multicolumn{2}{c}{IP\_Bot}&\multicolumn{2}{c}{DpIP\_Bot}\\
&&\multicolumn{1}{c}{avg}&t.o.&\multicolumn{1}{c}{avg}&t.o.&\multicolumn{1}{c}{avg}&t.o.&\multicolumn{1}{c}{avg}&t.o.&\multicolumn{1}{c}{avg}&t.o.&\multicolumn{1}{c}{avg}&t.o.\\
\midrule
50 - 54&100&5.978&0&0.136&0&1.970&0&0.140&0&0.568&0&0.131&0\\
55 - 59&100&10.208&7&0.136&0&2.922&0&0.141&0&0.764&0&0.132&0\\
60 - 64&100&13.791&31&0.141&0&5.245&0&0.145&0&1.076&0&0.134&0\\
65 - 69&100&18.372&57&0.144&0&6.562&1&0.148&0&1.390&0&0.135&0\\
70 - 74&100&20.195&75&0.146&0&8.513&15&0.151&0&1.856&0&0.137&0\\
75 - 79&100&22.485&87&0.149&0&11.003&10&0.154&0&2.372&0&0.138&0\\
80 - 84&100&22.865&91&0.150&0&12.489&18&0.157&0&3.031&0&0.139&0\\
85 - 89&100&26.028&94&0.154&0&14.864&25&0.160&0&3.746&0&0.140&0\\
90 - 94&100&26.866&98&0.158&0&17.244&48&0.167&0&4.861&0&0.144&0\\
95 - 99&100&0&100&0.160&0&18.644&57&0.170&0&5.808&0&0.147&0\\
\bottomrule
\end{tabularx}
\label{sub3seg}
\end{table}
\begin{table}[H]
\setlength{\tabcolsep}{3pt}
\footnotesize
\centering
\caption{Experimental results with SUBLOG49 for edit distance}
\begin{tabular}{c@{\quad}r@{\quad}r@{\quad}r@{\quad}r@{\quad}r@{\quad}r@{\quad}r}
\toprule
\multirow{2}{*}{\# of nodes}&\multirow{2}{*}{\# of instances}&\multicolumn{2}{c}{IP\_Edit}&\multicolumn{2}{c}{DpIP\_Edit}&\multicolumn{2}{c}{IP\_DpClique\_E}\\
&&avg&t.o.&avg&t.o.&avg&t.o.\\
\midrule
50 - 54&100&1.275&0&0.263&0&1.643&0\\
55 - 59&100&2.323&0&0.317&0&3.014&0\\
60 - 64&100&4.032&0&0.395&0&5.452&3\\
65 - 69&100&4.756&0&0.402&0&6.721&6\\
70 - 74&100&6.231&1&0.450&0&7.188&10\\
75 - 79&100&8.808&10&0.567&0&9.787&19\\
80 - 84&100&11.850&6&0.583&0&10.037&28\\
85 - 89&100&12.429&21&0.665&0&10.145&34\\
90 - 94&100&13.595&33&0.678&0&11.228&34\\
95 - 99&100&15.711&30&0.829&0&12.084&39\\
\bottomrule
\end{tabular}
\label{sub49ted}
\end{table}
\begin{table}[H]
\setlength{\tabcolsep}{3pt}
\footnotesize
\centering
\caption{Experimental results with SUBLOG49 for segmental distance, bottom-up segmental distance and bottom-up distance}
\begin{tabularx}{\textwidth}{crrr@{\quad}rr@{\quad}rr@{\quad}rr@{\quad}rr@{\quad}rr}
\toprule
\multirow{2}{*}{\# of nodes}&\multirow{2}{*}{\begin{minipage}{0.45in}\# of \\instances\end{minipage}}&\multicolumn{2}{c}{IP\_Sg}&\multicolumn{2}{c}{DpIP\_Sg}&\multicolumn{2}{c}{IP\_BotSg}&\multicolumn{2}{c}{DpIP\_BotSg}&\multicolumn{2}{c}{IP\_Bot}&\multicolumn{2}{c}{DpIP\_Bot}\\
&&\multicolumn{1}{c}{avg}&t.o.&\multicolumn{1}{c}{avg}&t.o.&\multicolumn{1}{c}{avg}&t.o.&\multicolumn{1}{c}{avg}&t.o.&\multicolumn{1}{c}{avg}&t.o.&\multicolumn{1}{c}{avg}&t.o.\\
\midrule
50 - 54&100&2.130&0&0.143&0&0.739&0&0.142&0&0.376&0&0.130&0\\
55 - 59&100&4.704&0&0.147&0&1.521&0&0.145&0&0.514&0&0.133&0\\
60 - 64&100&6.795&11&0.151&0&2.863&3&0.150&0&0.707&0&0.153&0\\
65 - 69&100&7.741&8&0.162&0&2.544&1&0.154&0&0.830&0&0.135&0\\
70 - 74&100&9.277&19&0.158&0&3.257&2&0.159&0&1.036&0&0.139&0\\
75 - 79&100&12.421&38&0.162&0&5.143&6&0.162&0&1.376&0&0.139&0\\
80 - 84&100&12.707&39&0.167&0&5.788&7&0.169&0&1.644&0&0.142&0\\
85 - 89&100&14.817&46&0.170&0&7.136&3&0.176&0&2.129&0&0.144&0\\
90 - 94&100&13.267&65&0.175&0&8.479&8&0.179&0&2.361&0&0.147&0\\
95 - 99&100&16.752&65&0.181&0&8.776&16&0.184&0&2.881&0&0.148&0\\
\bottomrule
\end{tabularx}
\label{sub49seg}
\end{table}
We can observe that the proposed methods (DpIP\_Edit, DpIP\_Sg, DpIP\_BotSg, and DpIP\_Bot) show remarkably improved the previous methods (IP\_Edit, IP\_Sg, IP\_BotSg, and IP\_Bot) as most of instances are computed within 0.2 seconds. In order to measure the scalability of the proposed methods, we used the wide range of dataset. We selected input tree pairs so that the number of total nodes ranges from around 0 to around 850. The results are shown in Fig. \ref{plot}. For segmemtanl distance and bottom-up segmental distance, the smallest instance which exceeds our time limit of 30 seconds appears when the total number of nodes belongs to range 450 - 500 whereas it appears for tree edit distance when the number of nodes belongs to range 150 - 200. For bottom-up distance, all instances selected in this experiments are solved within 7 seconds.
\begin{figure}[H]
\includegraphics[width=\textwidth]{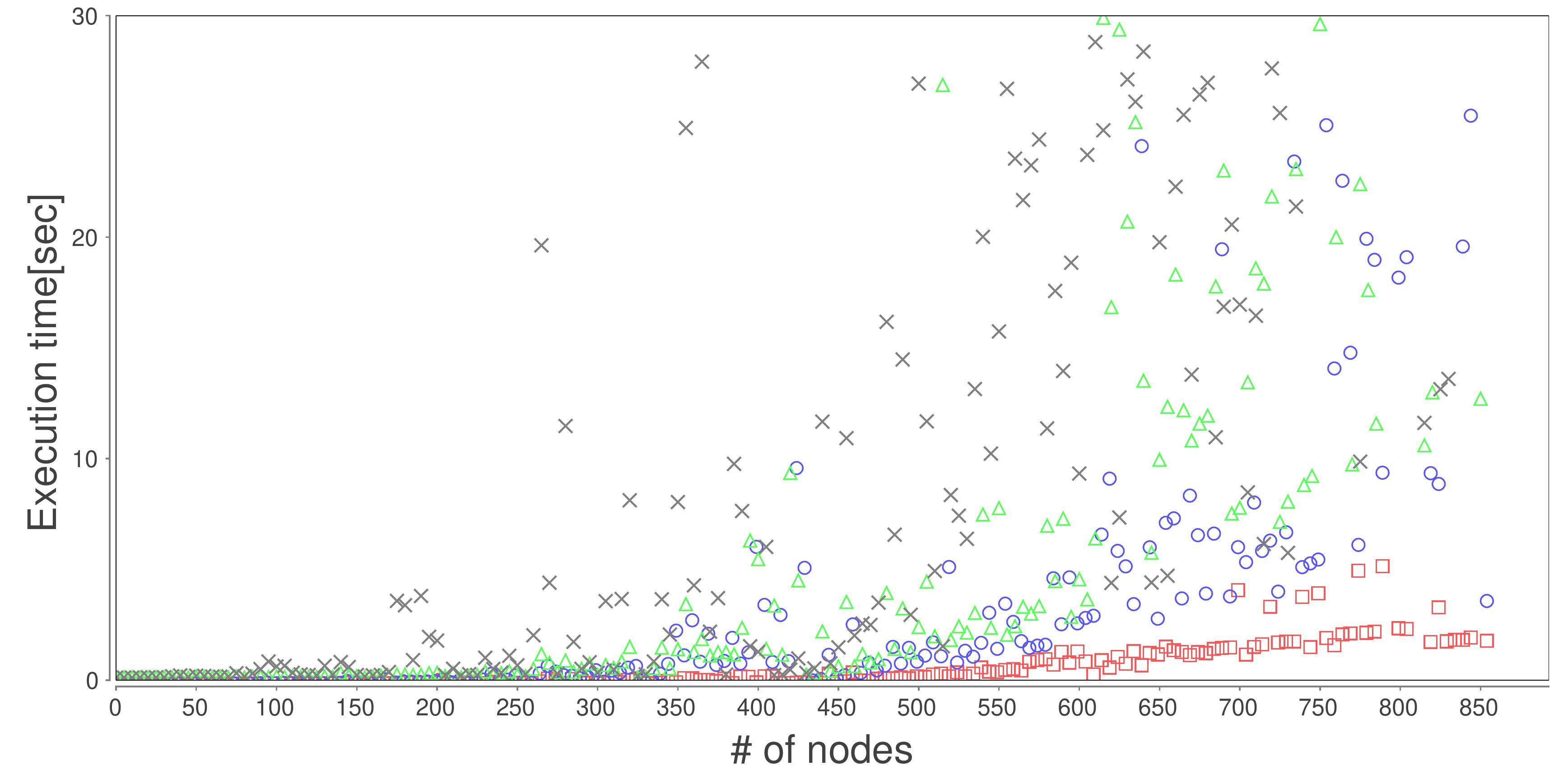}
\caption{The crosses, triangles, circles and squares represent the instances of the edit distance, segmental distance, bottom-up distance, and bottom-up distance problem, respectively.}
\label{plot}
\end{figure}
\section{Conclusion and Discussion}
We have proposed improved methods for computing the tree edit distance and its variants. While the naive IP formulation proposed by Kondo et al.~\cite{kondo} has $O(n^2m^2)$ constraints, our efficient IP formulation, though it has $O(nm)$ subproblems, only has $O(n+m)$ constraints. In case of segmental distance, bottom-up segmental distance and bottom-up distance, each subproblem, except for the problem combining the solutions of subproblems, can be reduced to the maximum weighted matching problem in a bipartite graph, which can be solved in polynomial time.

We performed some experiments using real tree-structured dataset. While the previous method only works for small-sized trees, our methods are still effective for large-sized trees.
In particular, for segmental distance and bottom-up segmental distance, our methods are available for trees whose total size is up to 450, and for bottom-up distance, every instance is solved within 7 seconds.

An advantage of IP-based method is that we can easily give an IP fomulation for another distance by adding some constraints to the IP formulation for edit distance. 
Therefore, extending our method to another important distance measure between unordered trees such as tree alignment distance \cite{alignment} would be our future work.
It would be interesting to develop practical algorithms for computing those distances without using general purpose solvers such as IP solvers or SAT solvers.

\bibliographystyle{splncs}

\begin{thebibliography}{10}
\providecommand{\url}[1]{\texttt{#1}}
\providecommand{\urlprefix}{URL }

\bibitem{akutsuapprox}
Akutsu, T., Fukagawa, D., Halldorsson, M.M., Takasu, A., Tanaka, K.:
  {Approximation and parameterized algorithms for common subtrees and edit
  distance between unordered trees}. Theoretical Computer Science  470,  10--22
  (2013)

\bibitem{akutsufpt}
Akutsu, T., Fukagawa, D., Takasu, A., Tamura, T.: {Exact algorithms for
  computing the tree edit distance between unordered trees}. Theoretical
  Computer Science  412(4-5),  352--364 (2011)

\bibitem{akutsu}
Akutsu, T., Tamura, T., Fukagawa, D., Takasu, A.: {Efficient exponential-time
  algorithms for edit distance between unordered trees}. Journal of Discrete
  Algorithms  25,  79--93 (2014)

\bibitem{demaine}
Demaine, E.D., Mozes, S., Rossman, B., Weimann, O.: {An optimal decomposition
  algorithm for tree edit distance}. {ACM} Transactions on Algorithms  6(1),
  1--19 (2009)

\bibitem{fukagawa}
Fukagawa, D., Tamura, T., Takasu, A., Tomita, E., Akutsu, T.: {A clique-based
  method for the edit distance between unordered trees and its application to
  analysis of glycan structures}. {BMC} Bioinformatics  12(Suppl 1),  S13
  (2011)

\bibitem{higuchi}
Higuchi, S., Kan, T., Yamamoto, Y., Hirata, K.: {An A* Algorithm for Computing
  Edit Distance between Rooted Labeled Unordered Trees}. In: New Frontiers in
  Artificial Intelligence, pp. 186--196. Springer Berlin Heidelberg (2012)

\bibitem{horesh}
Horesh, Y., Mehr, R., Unger, R.: {Designing an A* Algorithm for Calculating
  Edit Distance between Rooted-Unordered Trees}. Journal of Computational
  Biology  13(6),  1165--1176 (2006)

\bibitem{alignment}
Jiang, T., Wang, L., Zhang, K.: Alignment of trees {\textemdash} an alternative
  to tree edit. Theoretical Computer Science  143(1),  137--148 (1995)

\bibitem{seg}
Kan, T., Higuchi, S., Hirata, K.: {Segmental Mapping and Distance for Rooted
  Labeled Ordered Trees}. In: Algorithms and Computation, pp. 485--494.
  Springer Berlin Heidelberg (2012)

\bibitem{kegg}
Kanehisa, M., Goto, S.: {{KEGG}: Kyoto Encyclopedia of Genes and Genomes}.
  Nucleic Acids Research  28(1),  27--30 (2000)

\bibitem{kondo}
Kondo, S., Otaki, K., Ikeda, M., Yamamoto, A.: {Fast Computation of the Tree
  Edit Distance between Unordered Trees Using {IP} Solvers}. In: Discovery
  Science, pp. 156--167. Springer International Publishing (2014)

\bibitem{kuboyama}
Kuboyama, T.: Matching and Learning in Trees. Ph.D. thesis, The University of
  Tokyo (2007)

\bibitem{hungarian}
Kuhn, H.W.: {The Hungarian method for the assignment problem}. Naval Research
  Logistics Quarterly  2(1-2),  83--97 (1955)

\bibitem{mori}
Mori, T., Tamura, T., Fukagawa, D., Takasu, A., Tomita, E., Akutsu, T.: {A
  Clique-Based Method Using Dynamic Programming for Computing Edit Distance
  Between Unordered Trees}. Journal of Computational Biology  19(10),
  1089--1104 (2012)

\bibitem{tomita}
Nakamura, T., Tomita, E.: {Efficient algorithms for finding a maximum clique
  with maximum vertex weight (in Japanese)}. {Technical Report}, the University
  of Electro-Communications (2005)

\bibitem{tai}
Tai, K.C.: {The Tree-to-Tree Correction Problem}. Journal of the {ACM}  26(3),
  422--433 (1979)

\bibitem{bot}
Valiente, G.: {An efficient bottom-up distance between trees}. In: Proceedings
  Eighth Symposium on String Processing and Information Retrieval. {IEEE}
  (2001)

\bibitem{zaki}
Zaki, M.: {Efficiently mining frequent trees in a forest: algorithms and
  applications}. {IEEE} Transactions on Knowledge and Data Engineering  17(8),
  1021--1035 (2005)

\bibitem{ted-max-snp-hard}
Zhang, K., Jiang, T.: {Some {MAX} {SNP}-hard results concerning unordered
  labeled trees}. Information Processing Letters  49(5),  249--254 (1994)

\bibitem{ted-np-hard}
Zhang, K., Statman, R., Shasha, D.: {On the editing distance between unordered
  labeled trees}. Information Processing Letters  42(3),  133--139 (1992)

\end{thebibliography}

\end{document}